\documentclass{article}
\usepackage[T1]{fontenc}
\usepackage[english]{babel}
\usepackage{csquotes}
\usepackage{graphicx}
\usepackage{amsmath}
\usepackage{amssymb}
\usepackage{array}
\usepackage{float}
\usepackage{amsthm}
\usepackage{tikz}
\usetikzlibrary{calc, arrows.meta}
\usetikzlibrary{positioning}
\usetikzlibrary{graphs}
\usepackage[a4paper]{geometry}
\definecolor{csfInk}{HTML}{243447}
\definecolor{csfBlue}{HTML}{2F6690}
\definecolor{csfTeal}{HTML}{1B998B}
\definecolor{csfGreen}{HTML}{4E9F3D}
\definecolor{csfRed}{HTML}{B23A48}
\definecolor{csfGray}{HTML}{E8ECEF}
\tikzset{
	csfblock/.style={
		draw=csfInk!55,
		rounded corners=2pt,
		fill=white,
		line width=0.45pt,
		align=center,
		inner xsep=7pt,
		inner ysep=6pt,
		font=\small
	},
	csfarrow/.style={-{Latex[length=2.4mm]}, line width=0.55pt, draw=csfInk!70},
	state/.style={
		circle,
		draw=csfInk!60,
		fill=white,
		minimum size=8mm,
		inner sep=0pt,
		font=\small\ttfamily
	},
	attractorstate/.style={
		state,
		draw=csfTeal!85!black,
		fill=csfTeal!12,
		line width=0.75pt
	},
	fixedstate/.style={
		state,
		draw=csfGreen!75!black,
		fill=csfGreen!14,
		line width=0.75pt
	},
	transientstate/.style={
		state,
		draw=csfInk!35,
		fill=csfGray!45
	}
}
\newtheorem{theorem}{Theorem}
\newtheorem{lemma}[theorem]{Lemma}
\theoremstyle{definition}
\newtheorem{definition}[theorem]{Definition}
\newtheorem{problem}[theorem]{Problem}

\usepackage[backend=bibtex,style=numeric,sorting=none]{biblatex}
\newcommand{\classP}{\mathrm{P}}
\newcommand{\classNP}{\mathrm{NP}}
\newcommand{\classPSPACE}{\mathrm{PSPACE}}
\newcommand{\MajClass}{\mathsf{Maj}}
\newcommand{\OrAndClass}{\mathsf{OrAnd}}
\newcommand{\AndOrClass}{\mathsf{AndOr}}
\newcommand{\AffClass}{\mathsf{Aff}}
\newcommand{\ConjClass}{\mathsf{Conj}}
\newcommand{\DisjClass}{\mathsf{Disj}}
\newcommand{\msuaffiliation}{AI Research Center, Faculty of Mechanics and Mathematics, Lomonosov Moscow State University}
\newcommand{\drobyshevemail}{drobyshev.sanya@yandex.ru}
\newcommand{\bokovemail}{grigoriy.bokov@math.msu.ru}

\title{A Complete Complexity Dichotomy for Cyclic Attractor Detection in Boolean Networks with Restricted Local Rules}
\author{Alexander Drobyshev\textsuperscript{1}
	\and Grigoriy Bokov\textsuperscript{1,*}}
\date{}

\begin{document}

	\maketitle
	\begin{center}
		\vspace{-1.5em}
		{\small
		\textsuperscript{1}\msuaffiliation\par
		\textsuperscript{*}Correspondence and requests for materials should be
		addressed to G.B. (\bokovemail).\par
		E-mail addresses: A.D., \drobyshevemail; G.B., \bokovemail.}
	\end{center}
	\vspace{0.75em}

\begin{abstract}
	Boolean networks (BNs) are finite models of interacting systems whose
	long-term behaviour is organised by attractors. After a transient phase, every
	synchronously updated BN reaches either a fixed state or a cyclic attractor.
	Therefore, deciding whether a prescribed cyclic behaviour exists is a basic
	computational task in finite-state network dynamics. Here we study this task
	for networks with restricted local rules: each node is updated by a Boolean
	rule from a fixed closed rule family, and the period is fixed in advance. We
	prove a complete dichotomy for every fixed period at least two. For each closed
	rule family, cyclic-attractor detection is either polynomial-time decidable or
	\(\classNP\)-complete. The intractable cases are precisely those that can
	implement majority-like self-dual behaviour or mixed monotone
	conjunction-disjunction mechanisms. In the tractable cases, affine, purely
	conjunctive, or purely disjunctive structure reduces the problem to linear
	algebra or graph reachability. These criteria identify which
	restricted local rule families preserve computational tractability in finite
	Boolean network models.
\end{abstract}


	\section{Introduction}

	A Boolean network (BN) is a collection of binary-valued elements, or nodes,
	which influence one another according to local Boolean rules. The state of
	each node is represented by \(1\) or \(0\). At each discrete time step, the
	node values are updated either synchronously or asynchronously. Here we focus
	on the synchronous update scheme, in which all node values are
	updated simultaneously. A BN with \(n\) nodes is specified by \(n\) Boolean
	functions, one for each coordinate, and its global state is a binary vector in
	\(\{0,1\}^n\). Thus the dynamics form a finite map from \(\{0,1\}^n\) to
	itself.

	Since their introduction, BNs have been used as finite models of
	interacting systems. McCulloch and Pitts introduced a Boolean description of
	neural activity \cite{mcculloch}, and Kauffman later used Boolean networks as
	models of genetic control \cite{kauffman}. Related logical descriptions became
	standard in biochemical regulation through the work of Thomas and of Glass and
	Kauffman \cite{thomas,glassKauffman1973}. BNs also appear in neural
	computation \cite{hopfield}, coding and network-coding problems
	\cite{riis,gadouleau2011,gadouleau}, and social dynamics \cite{poindron}. In
	these applications, the same mathematical object is used to represent how
	local binary rules generate global finite-state behaviour.

	Regardless of the initial state, the trajectory of a finite deterministic BN
	must eventually enter a state or a closed cycle of states from which it cannot
	escape. These recurrent objects are the attractors of the BN. A cycle of
	length one is a fixed-point attractor, whereas a cycle of length at least two
	is a cyclic attractor. In neural models, attractors have been used to describe
	memory states and transitions between recurrent patterns
	\cite{goles,chossat,ladwani}. In regulatory models, attractors are used to
	represent stable cellular states and phenotypes
	\cite{kauffman,thomas}. Oscillatory regimes are associated with recurrent
	attractors in logical descriptions of biochemical control
	\cite{glassKauffman1973}, and negative feedback in circular gene networks
	provides a concrete mechanism for cyclic behaviour in biochemical models
	\cite{likhoshvai2020}. Our contribution is formal rather than
	application-specific: given a BN and a prescribed period, we study the
	decision problem of whether the network has a cyclic attractor of exactly that
	period.

	Figure~\ref{fig:attractor-landscape} illustrates these basic objects for a
	small three-variable BN under synchronous update. The local rules define a
	global map on eight configurations. Its state-transition graph contains two
	fixed-point attractors and one exact period-two cyclic attractor. Although the
	example is small, it illustrates the distinction used throughout the paper:
	recurrence is a property of the global finite map, while the model itself is
	specified locally through coordinate functions.

	\begin{figure}[H]
		\centering
		\resizebox{\textwidth}{!}{%
		\begin{tikzpicture}[
			x=1cm,y=1cm,
			landscapePanel/.style={
				draw=csfInk!18,
				rounded corners=4pt,
				line width=0.45pt
			},
			basinPanel/.style={
				draw=csfInk!13,
				rounded corners=3pt,
				line width=0.35pt
			}
			]
			\path[landscapePanel, fill=csfBlue!5] (-0.5,-2.25) rectangle (3.4,2.25);
			\path[landscapePanel, fill=csfGray!35] (3.8,-2.25) rectangle (12.7,2.25);

			\node[font=\bfseries\small, text=csfInk] at (1.45,1.95) {(a) Local rules};
			\node[font=\bfseries\small, text=csfInk] at (8.25,1.95) {(b) State-transition landscape};

			\node[state, fill=csfBlue!10] (x1) at (0.65,0.55) {$x_1$};
			\node[state, fill=csfBlue!10] (x2) at (2.25,0.55) {$x_2$};
			\node[state, fill=csfBlue!10] (x3) at (1.45,1.28) {$x_3$};
			\draw[csfarrow] (x1) to[bend left=18] (x2);
			\draw[csfarrow] (x2) to[bend left=18] (x1);
			\draw[csfarrow] (x1) -- (x3);
			\draw[csfarrow] (x2) -- (x3);
			\node[csfblock, fill=white, text width=2.9cm] at (1.45,-1.02)
				{\(f_1=x_2\)\\ \(f_2=x_1\)\\ \(f_3=x_1\vee x_2\)};

			\path[basinPanel, fill=csfGreen!9] (4.2,-1.65) rectangle (5.75,1.45);
			\path[basinPanel, fill=csfTeal!9] (6.1,-1.65) rectangle (10.4,1.45);
			\path[basinPanel, fill=csfGreen!9] (10.85,-1.65) rectangle (12.35,1.45);
			\node[fixedstate] (s000) at (4.98,0.72) {000};
			\node[transientstate] (s001) at (4.98,-0.62) {001};
			\node[attractorstate] (s011) at (7.15,0.72) {011};
			\node[attractorstate] (s101) at (9.35,0.72) {101};
			\node[transientstate] (s100) at (7.15,-0.72) {100};
			\node[transientstate] (s010) at (9.35,-0.72) {010};
			\node[fixedstate] (s111) at (11.6,0.72) {111};
			\node[transientstate] (s110) at (11.6,-0.62) {110};
			\draw[csfarrow] (s000) edge[loop above, looseness=5] (s000);
			\draw[csfarrow] (s001) -- (s000);
			\draw[csfarrow] (s100) -- (s011);
			\draw[csfarrow] (s010) -- (s101);
			\draw[csfarrow] (s011) to[out=25,in=155] (s101);
			\draw[csfarrow] (s101) to[out=-155,in=-25] (s011);
			\draw[csfarrow] (s110) -- (s111);
			\draw[csfarrow] (s111) edge[loop above, looseness=5] (s111);
			\node[font=\scriptsize, text=csfGreen!55!black] at (4.98,-1.32) {fixed basin};
			\node[font=\scriptsize, text=csfTeal!65!black] at (8.25,-1.32) {period-\(2\) basin};
			\node[font=\scriptsize, text=csfGreen!55!black] at (11.6,-1.32) {fixed basin};
		\end{tikzpicture}}
		\caption{From local Boolean rules to the attractor landscape for
		\(f(x_1,x_2,x_3)=\left(x_2,x_1,x_1\vee x_2\right)\). Panel (a) shows the
		three-variable Boolean network, and panel (b) shows its full finite
		state-transition graph with two fixed-point attractors and one cyclic
		attractor of exact period \(2\).}
		\label{fig:attractor-landscape}
	\end{figure}
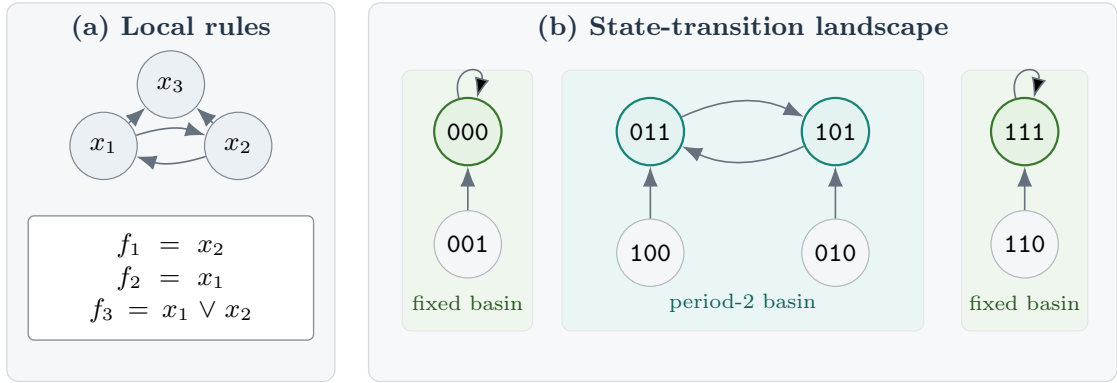

	A central question for BNs and related finite network models is how local
	interactions shape global recurrence. Several lines of work explain
	attractors through feedback in the interaction graph. Positive
	feedback circuits are linked to multistationarity and memory-like behaviour
	\cite{thomasRichelle1988,thomasKaufman2001}. Negative circuits are associated
	with sustained oscillation and cyclic recurrence
	\cite{richard2010negative,richard2019cycles}. Later refinements relate signed
	cycles, monotone networks, and local negative circuits to the possible number
	and structure of fixed points and cyclic attractors
	\cite{aracena2017fixedCycles,tonello2019negativeCircuits}. For asynchronous
	Boolean dynamics, signed-cycle conditions can also separate attractors
	\cite{richardTonello2023Separation}. These results show that graph structure
	and signs constrain the possible recurrent behaviour.

	Random Boolean networks give another reason to separate local rule form from
	network wiring. In Kauffman-type ensembles, connectivity, bias, and sensitivity
	shape whether the dynamics are ordered, critical, or chaotic; the annealed
	approximation of Derrida and Pomeau gives a classical account of this
	dependence \cite{derridaPomeau1986}, and later work developed the
	random-network perspective further \cite{aldana2003boolean,drossel2008random}.
	Attractor counts and stability also vary with local-rule statistics
	\cite{samuelssonTroein2003,drossel2005number,klemmBornholdt2005}. More
	structured studies make the same point for particular logical formalisms:
	damage spreading, Lyapunov spectra, and regulatory logic classes relate update
	logic to dynamical regimes \cite{vispoel2024damage,sil2025logicClass}, while
	canalization and threshold-rule analyses show that modelling choices at the
	level of local rules can affect regulatory dynamics
	\cite{kadelkaMurrugarra2024Canalization,kadelkaHari2025ThresholdModels}. These
	results motivate local-rule restrictions, but they are probabilistic,
	empirical, or model-specific rather than a worst-case classification for
	explicitly specified finite Boolean maps.

	For finite Boolean networks themselves, attractor analysis has also been
	studied as an algorithmic problem. The full state-transition graph has
	\(2^n\) vertices, so direct enumeration is usually infeasible. Algorithms have
	been proposed for small attractors \cite{zhang2007smallAttractors}, SAT-based
	detection in synchronous Boolean networks \cite{dubrovaTeslenko2011}, and
	attractor computation in both synchronous and asynchronous Boolean networks
	\cite{zheng2013attractors}. Answer Set Programming gives another logical
	encoding for attractor computation in gene-regulatory models
	\cite{khaled2023aspBooleanNetworks}. Periodic-attractor search can also use
	a priori information about the network \cite{munzner2022PeriodicAttractors},
	and more general attractor detection and enumeration algorithms have been
	developed \cite{moriAkutsu2022Algorithms}. Algebraic state-space methods give
	another structural description of fixed points and limit cycles,
	including hidden order induced by dual Boolean networks
	\cite{zhangJiCheng2024HiddenOrder}.

	Existing worst-case results give two reference points. If the period is part
	of the input, limit-cycle existence is \(\classNP\)-complete \cite{bridoux}.
	If the period is one, attractor detection is fixed-point existence. Complete
	dichotomies are known for fixed points of Boolean dynamical systems with local
	functions restricted by closed Boolean-function classes~\cite{post1941} and
	interaction graphs restricted by minor-closed graph classes
	\cite{kosubFixedPoints}; counting versions have analogous dichotomies
	\cite{homanKosub2015}. Related work asks which signed interaction digraphs can
	realise prescribed fixed-point behaviour \cite{bridouxMaxFixedPoints}, and
	recent SAT-based methods show that singleton attractor computation remains an
	active algorithmic problem
	\cite{higuchi2025SingletonAttractors}.

	The fixed-period cyclic case lies between these settings. The period is no
	longer part of the input, so intractability cannot be attributed to an
	unbounded temporal specification. At the same time, the attractor is
	non-stationary, so fixed-point dichotomies do not determine the answer. We
	therefore ask whether the closed family of allowed local update rules alone
	determines the worst-case complexity of detecting a cyclic attractor of a
	prescribed exact period.

	Here we answer this question for all closed families of local Boolean rules.
	The update scheme is synchronous, the interaction graph is
	unrestricted, and the period is fixed in advance. Under these assumptions, we
	prove a complete dichotomy: for every fixed period at least two, each closed
	local rule family yields either a polynomial-time decision problem or an
	\(\classNP\)-complete one. The intractable cases are exactly those whose local
	rules can express majority-like self-dual behaviour or mixed monotone
	conjunction-disjunction mechanisms. The tractable cases are explained by
	affine, purely conjunctive, or purely disjunctive structure.

	The classification separates the computational effects of different
	restricted local rules. Some restricted rule families can still encode a
	global consistency test, whereas others preserve algebraic or one-sided order
	structure that makes exact-period detection polynomial-time decidable. The
	Results state this boundary and illustrate the corresponding dynamical
	mechanisms; the Methods give the representation model, algorithms,
	reductions, and full proofs.

	\section{Preliminaries}

	\subsection{Boolean networks and exact periods}

	This section fixes the notation and definitions used in the classification.
	We first describe Boolean networks as deterministic finite dynamical systems
	and define exact-period attractors. We then specify the input representation
	and the decision problem studied in the Results and Methods sections.

	Let \(\mathbb B=\{0,1\}\) denote the set of binary values. For a BN with
	\(n\) nodes, the state space is \(\mathbb B^n\). A state is written as
	\(x=(x_1,\ldots,x_n)\), where \(x_i\in\mathbb B\) is the value of node \(i\).

	For a positive integer \(n\), we write
	\[
	[n]=\{1,\ldots,n\}.
	\]

	Let \(f_1,\ldots,f_n\) be Boolean functions
	\[
	f_i:\mathbb B^n\to\mathbb B, \qquad i\in[n].
	\]

	\begin{definition}[Boolean network and interaction digraph]
		A \emph{Boolean network} on \(n\) components is a mapping
		\[
		f:\mathbb B^n\to\mathbb B^n,
		\qquad
		f=(f_1,\ldots,f_n),
		\]
		where each \(f_i\) is a Boolean function. The set \(\mathbb B^n\) is the
		configuration space.
		Following~\cite{bridoux}, the \emph{interaction digraph} of \(f\) is
		\(G_f=(V,E)\), where \(V=[n]\) and \(E\) is defined by
		\[
		(i,j)\in E
		\quad\Longleftrightarrow\quad
		f_j \text{ essentially depends on } x_i.
		\]
	\end{definition}

	Throughout this article the update scheme is synchronous: one step sends a state
	\(x\) to \(f(x)\). The map \(f\) induces a directed state-transition graph whose
	vertices are the configurations in \(\mathbb B^n\). The outgoing edge from a state
	\(x\) points to \(f(x)\). Since this graph is finite and deterministic,
	every trajectory eventually enters a directed cycle. These recurrent cycles
	are the attractors of the finite dynamical system.

	\begin{definition}[Attractors of exact period]
		A configuration \(x\in\mathbb B^n\) has \emph{exact period \(k\)} under \(f\) if
		\[
		f^k(x)=x
		\]
		and
		\[
		f^\ell(x)\neq x
		\qquad
		\text{for all } \ell\in\{1,\ldots,k-1\}.
		\]
		The corresponding orbit
		\[
		x,\ f(x),\ \ldots,\ f^{k-1}(x)
		\]
		is a \emph{cyclic attractor of exact period \(k\)} when \(k\ge 2\). For
		\(k=1\), it is a \emph{fixed-point attractor}. The set of configurations
		of exact period \(k\) is denoted by
		\[
		\Phi_k(f)=
		\left\{
		x\in\mathbb B^n
		\;\middle|\;
		f^k(x)=x
		\text{ and }
		\forall \ell\in\{1,\ldots,k-1\},\ f^\ell(x)\neq x
		\right\}.
		\]
		In the Boolean-network complexity literature, such cyclic attractors are
		often called limit cycles.
	\end{definition}

	Cell-cycle regulation provides a compact setting in which Boolean update rules
	and recurrent behaviour can be compared directly. In Boolean descriptions of
	yeast cell-cycle control, a regulatory component is encoded by a binary
	variable, with value \(1\) meaning active and value \(0\) meaning inactive;
	recurrent states describe regulatory patterns that persist or repeat over time
	\cite{li2004YeastCellCycle,davidichBornholdt2008FissionYeast}.
	Figure~\ref{fig:cell-cycle-bn} displays a four-node synchronous network in this
	spirit. The variables indicate a growth condition, cyclin activity, and two
	cyclin-degradation regulators. The rule
	\(f_2(x)=x_1\wedge\neg x_4\) expresses a simple regulatory dependence: cyclin
	can be active only when growth is present and Cdh1 is inactive. Because the
	network is small, its complete state-transition graph can also be displayed.
	The finite map has three attractors. Two have natural readings in the
	cell-cycle language. The fixed point \(0000\) represents absence of growth and
	inactive downstream components, and all configurations with \(x_1=0\) flow to
	this state. The exact period-six attractor follows the usual activation order
	under sustained growth: cyclin activity precedes Cdc20 and Cdh1 activation,
	after which cyclin is removed and the cycle returns to its initial state.

	\begin{figure}[H]
		\centering
		\resizebox{0.98\textwidth}{!}{%
		\begin{tikzpicture}[
			x=1cm,y=1cm,
			panel/.style={
				draw=csfInk!16,
				rounded corners=4pt,
				line width=0.45pt
			},
			subpanel/.style={
				draw=csfInk!14,
				fill=white,
				rounded corners=3pt,
				line width=0.38pt
			},
			networkNode/.style={
				circle,
				draw=csfInk!48,
				fill=white,
				line width=0.55pt,
				minimum size=11.3mm,
				align=center,
				inner sep=1.2pt,
				font=\scriptsize
			},
			rulePanel/.style={
				align=center,
				inner sep=0pt,
				font=\scriptsize
			},
			stateNode/.style={
				rectangle,
				rounded corners=1.5pt,
				draw=csfInk!38,
				fill=white,
				line width=0.36pt,
				minimum width=6.7mm,
				minimum height=4.35mm,
				inner xsep=0.6pt,
				inner ysep=0.35pt,
				font=\tiny\ttfamily
			},
			cycleState/.style={
				stateNode,
				draw=csfTeal!78!black,
				fill=csfTeal!10,
				line width=0.66pt
			},
			altCycleState/.style={
				stateNode,
				draw=csfBlue!78!black,
				fill=csfBlue!8,
				line width=0.58pt
			},
			fixedState/.style={
				stateNode,
				draw=csfInk!35,
				fill=csfGray!45,
				line width=0.52pt
			},
			graphArrow/.style={-{Latex[length=1.6mm]}, line width=0.48pt, draw=csfInk!54, shorten <=1.15pt, shorten >=1.15pt},
			depArrow/.style={-{Latex[length=2.8mm]}, line width=0.95pt, draw=csfInk!67, line cap=round, line join=round, shorten <=2.8pt, shorten >=2.8pt},
			negArrow/.style={-{Latex[length=2.8mm]}, line width=1.12pt, draw=csfRed!82!black, line cap=round, line join=round, shorten <=3pt, shorten >=3pt},
			cycleArrow/.style={-{Latex[length=1.85mm]}, line width=0.66pt, draw=csfTeal!78!black, shorten <=1.5pt, shorten >=1.5pt},
			altCycleArrow/.style={-{Latex[length=1.75mm]}, line width=0.58pt, draw=csfBlue!78!black, shorten <=1.5pt, shorten >=1.5pt}
			]

			\fill[panel, fill=csfBlue!5] (-0.35,-3.10) rectangle (16.55,2.62);
			\path[subpanel] (0.05,-2.76) rectangle (6.65,2.28);
			\path[subpanel] (6.80,-2.76) rectangle (16.25,2.28);

			\node[font=\bfseries\small, text=csfInk] at (3.35,1.87)
				{Boolean network};
			\node[font=\bfseries\small, text=csfInk] at (11.53,1.87)
				{State-transition graph};

			\node[networkNode, fill=csfBlue!8] (growth) at (1.32,0.43)
				{Growth\\[-0.15em]\(\scriptstyle x_1\)};
			\node[networkNode, fill=csfTeal!9] (clb) at (3.35,0.43)
				{Cyclin\\[-0.15em]\(\scriptstyle x_2\)};
			\node[networkNode, fill=csfTeal!9] (cdc20) at (5.48,0.43)
				{Cdc20\\[-0.15em]\(\scriptstyle x_3\)};
			\node[networkNode, fill=csfTeal!9] (cdh1) at (4.26,-0.87)
				{Cdh1\\[-0.15em]\(\scriptstyle x_4\)};
			\draw[depArrow, draw=csfInk!48] (growth) edge[loop above, looseness=8, min distance=4.8mm] (growth);
			\draw[depArrow] (growth) -- (clb);
			\draw[depArrow] (clb) -- (cdc20);
			\draw[depArrow] (cdc20.south) to[out=-92,in=18,looseness=1.35] (cdh1.east);
			\draw[negArrow] (cdh1.west) to[out=180,in=-88] (clb.south);

			\node[rulePanel] at (3.35,-2.12)
				{\begin{tabular}{@{}l@{\hspace{1.05em}}l@{}}
				\(f_1(x)=x_1\) &
				\(f_2(x)=x_1\wedge\neg x_4\)\\
				\(f_3(x)=x_2\) &
				\(f_4(x)=x_3\)
				\end{tabular}};

			\node[font=\scriptsize, text=csfInk!70] at (13.92,1.24)
				{state order: \((x_1,x_2,x_3,x_4)\)};

			\node[fixedState] (s0000) at (8.14,0.33) {0000};
			\node[stateNode] (s0001) at (9.24,0.33) {0001};
			\node[stateNode] (s0010) at (10.22,0.65) {0010};
			\node[stateNode] (s0011) at (10.22,-0.01) {0011};
			\node[stateNode] (s0100) at (11.26,0.91) {0100};
			\node[stateNode] (s0101) at (11.26,0.43) {0101};
			\node[stateNode] (s0110) at (11.26,-0.19) {0110};
			\node[stateNode] (s0111) at (11.26,-0.67) {0111};

			\draw[graphArrow, draw=csfInk!38] (s0000) edge[loop left, looseness=8, min distance=3.8mm] (s0000);
			\draw[graphArrow] (s0001) -- (s0000);
			\draw[graphArrow] (s0010) -- (s0001);
			\draw[graphArrow] (s0011) -- (s0001);
			\draw[graphArrow] (s0100) -- (s0010);
			\draw[graphArrow] (s0101) -- (s0010);
			\draw[graphArrow] (s0110) -- (s0011);
			\draw[graphArrow] (s0111) -- (s0011);

			\node[cycleState] (s1000) at (13.90,0.52) {1000};
			\node[cycleState] (s1100) at (15.08,-0.01) {1100};
			\node[cycleState] (s1110) at (15.08,-0.95) {1110};
			\node[cycleState] (s1111) at (13.90,-1.48) {1111};
			\node[cycleState] (s1011) at (12.72,-0.95) {1011};
			\node[cycleState] (s1001) at (12.72,-0.01) {1001};
			\draw[cycleArrow] (s1000) -- (s1100);
			\draw[cycleArrow] (s1100) -- (s1110);
			\draw[cycleArrow] (s1110) -- (s1111);
			\draw[cycleArrow] (s1111) -- (s1011);
			\draw[cycleArrow] (s1011) -- (s1001);
			\draw[cycleArrow] (s1001) -- (s1000);
			\node[font=\scriptsize, text=csfTeal!72!black] at (13.90,-1.95)
				{exact period \(6\)};

			\node[font=\scriptsize, text=csfInk!64, align=center] at (8.14,-0.08)
				{fixed point};

			\node[altCycleState] (s1010) at (8.62,-1.48) {1010};
			\node[altCycleState] (s1101) at (9.78,-1.48) {1101};
			\draw[altCycleArrow] (s1010.north east) to[out=25,in=155] (s1101.north west);
			\draw[altCycleArrow] (s1101.south west) to[out=-155,in=-25] (s1010.south east);
			\node[font=\scriptsize, text=csfBlue!72!black] at (9.20,-1.95)
				{exact period \(2\)};
		\end{tikzpicture}}
		\caption{A cell-cycle example. The left
		panel shows the Boolean network and its coordinate functions; the right
		panel shows the full synchronous state-transition graph, with attractors
		highlighted.}
		\label{fig:cell-cycle-bn}
	\end{figure}
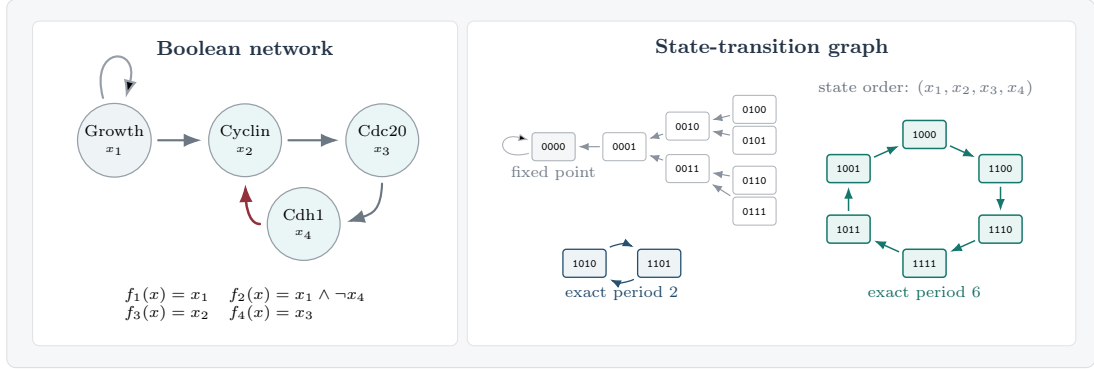

	\subsection{Decision problem and input encoding}

	The restrictions considered in this work are imposed on the coordinate
	functions \(f_i\), that is, on the local update rules of the network. They are
	expressed through closed classes of Boolean functions.
	\begin{definition}[Closure of a function class]
		Let \(\mathcal F\) be a set of Boolean functions. Its \emph{closure} under
		superposition and introduction of fictitious variables is denoted by
		\([\mathcal F]\).
	\end{definition}

	Post introduced the lattice of closed Boolean-function classes
	\cite{post1941}; complete inclusion diagrams and finite generating sets are
	given, for example, by Lau~\cite{lau2006}. In the decision problem below, a
	local rule family is fixed before the input is given. The input network is
	then represented by circuits over that fixed finite family. The detailed
	circuit convention is stated in the Methods, where it is needed to verify the
	size of reductions and the running time of algorithms. The classification
	depends on the closed class of admissible local rules, not on the particular
	finite generating set chosen to represent it.

	With this notation, the decision problem studied in this article is the
	following. We use the label \(k\)-PLCE for prescribed-period limit-cycle
	existence. The local rule family is fixed by the finite family
	\(\mathcal F\), while the network itself is given by circuits over
	\(\mathcal F\).
	In every complexity statement, the period \(k\) is also fixed in advance and is
	not part of the input. Thus the input size measures the explicit Boolean
	network representation, not the numerical encoding of the requested period.

	\begin{problem}[\(k\)-PLCE\((\mathcal F)\)]
		Fix an integer \(k\ge 1\) and a finite set \(\mathcal F\) of Boolean
		functions.
		The input is a Boolean network \(f=(f_1,\ldots,f_n)\) under the synchronous
		update scheme, where each coordinate \(f_i\in[\mathcal F]\) is represented
		by a circuit over the fixed family \(\mathcal F\). The question is whether
		\[
		\Phi_k(f)\neq\emptyset.
		\]
	\end{problem}

	When \(\mathcal C\) is a closed class with a fixed finite generating set
	\(\mathcal B_{\mathcal C}\), we write \(k\)-PLCE\((\mathcal C)\) for
	\(k\)-PLCE\((\mathcal B_{\mathcal C})\). This notation is independent of the
	particular finite generating set chosen for \(\mathcal C\), because fixed
	finite bases for the same closed class give polynomially equivalent circuit
	encodings. Dynamically, \(k\)-PLCE asks
	whether the attractor landscape contains a fixed point when \(k=1\), or a
	non-stationary recurrent pattern of prescribed exact period when \(k\ge 2\).

	\section{Results}

	\subsection{Complete local-rule classification}

	The central result is a complete classification of exact-period cyclic
	attractor detection for Boolean networks with restricted local rules.
	The interaction graph is unrestricted, and the update scheme is synchronous.
	Thus two networks may have comparable wiring, but the computational
	difficulty of analysing their attractor landscapes can change when the local
	rules are affine, monotone, self-dual, conjunctive, or disjunctive. The
	classification below gives the polynomial versus \(\classNP\)-complete
	boundary and then illustrates the dynamical mechanisms that appear there.

	Table~\ref{tab:boundary-classes} lists the six rule families that determine
	the classification. For the majority boundary, we use the odd-arity majority
	function
	\[
	\operatorname{maj}_{2q+1}(x_1,\ldots,x_{2q+1})=1
	\quad\Longleftrightarrow\quad
	x_1+\cdots+x_{2q+1}\ge q+1,
	\]
	where the sum is taken in \(\mathbb Z\). Thus
	\(\operatorname{maj}_3(x,y,z)=(x\wedge y)\vee(y\wedge z)\vee
	(x\wedge z)\). When the arity is clear, we write \(\operatorname{maj}\);
	repeated arguments are allowed.

	\begin{table}[H]
		\centering
		\begingroup
		\small
		\setlength{\tabcolsep}{3pt}
		\begin{tabular}{>{\raggedright\arraybackslash}p{0.24\textwidth}
			>{\centering\arraybackslash}p{0.13\textwidth}
			>{\raggedright\arraybackslash}p{0.21\textwidth}
			>{\raggedright\arraybackslash}p{0.32\textwidth}}
			\hline
			Rule family & Class label & Generated by & Local rule form \\
			\hline
			Majority self-dual & \(\MajClass\) & \(\operatorname{maj}_3\) & Monotone self-dual Boolean functions. \\
			Mixed monotone, OR-over-AND & \(\OrAndClass\) & \(x\vee(y\wedge z)\) & Monotone functions preserving \(0,1\) and bounded below by one argument. \\
			Mixed monotone, AND-over-OR & \(\AndOrClass\) & \(x\wedge(y\vee z)\) & Monotone functions preserving \(0,1\) and bounded above by one argument. \\
			Affine & \(\AffClass\) & \(1,\ x\oplus y\) & Functions \(c\oplus x_{i_1}\oplus\cdots\oplus x_{i_r}\). \\
			Conjunctive with constants & \(\ConjClass\) & \(x\wedge y,\ 0,\ 1\) & Constants and finite conjunctions of variables. \\
			Disjunctive with constants & \(\DisjClass\) & \(x\vee y,\ 0,\ 1\) & Constants and finite disjunctions of variables. \\
			\hline
		\end{tabular}
		\endgroup
		\caption{Boundary rule families for the local-rule dichotomy. The columns
		give the family name, the label used in Theorem~\ref{thm:post-dichotomy},
		a finite generating set, and the corresponding form of the local rules.}
		\label{tab:boundary-classes}
	\end{table}

	The first three rows of Table~\ref{tab:boundary-classes} are the hard
	mechanisms: any closed class containing one of them can express enough local
	logic to encode a global Boolean consistency test. The last three rows are the
	tractable mechanisms: they preserve affine or one-sided order structure under
	the iterates used for fixed-period detection. The theorem states that these
	two alternatives exhaust all closed Boolean-function classes for every fixed
	non-stationary period.

	\begin{theorem}[Local-rule dichotomy]\label{thm:post-dichotomy}
		Fix \(k\ge 2\). For every closed Boolean-function class
		\(\mathcal C\), \(k\)-PLCE\((\mathcal C)\) is either
		polynomial-time solvable or \(\classNP\)-complete. More precisely,
		\(k\)-PLCE\((\mathcal C)\) is \(\classNP\)-complete if \(\mathcal C\) contains at
		least one of \(\MajClass\), \(\OrAndClass\), or \(\AndOrClass\). In all
		remaining cases, \(\mathcal C\) is contained in \(\AffClass\),
		\(\ConjClass\), or \(\DisjClass\), and \(k\)-PLCE\((\mathcal C)\) is
		decidable in polynomial time.
	\end{theorem}

	The six named families are boundary cases, not the only classes covered by the
	theorem. Once the closed class generated by the admissible local rules is
	fixed, its position relative to Table~\ref{tab:boundary-classes} determines
	the complexity. Classes containing a hard mechanism are \(\classNP\)-complete;
	all other classes are covered by one of the affine, conjunctive, or
	disjunctive polynomial-time algorithms.

	The stationary case \(k=1\) gives a useful reference point. Then
	\(k\)-PLCE is fixed-point existence, and for unrestricted interaction graphs
	Kosub's dichotomy gives \(\classNP\)-completeness precisely for closed classes
	containing the self-dual class~\cite{kosubFixedPoints}. Theorem
	\ref{thm:post-dichotomy} shows that prescribing a non-stationary period changes
	the local-rule boundary: for every fixed \(k\ge2\), the hard classes are
	\(\MajClass\), \(\OrAndClass\), and \(\AndOrClass\).

	\subsection{Dynamical mechanisms at the boundary}

	The boundary is not determined by whether non-stationary cycles can occur:
	they occur on both sides of the dichotomy. The difference is how such cycles
	are constrained. In the tractable classes, the equations defining a prescribed
	period retain a rigid algebraic or order-theoretic form. In the hard classes,
	a local rule can make the propagation of a periodic signal conditional on a
	Boolean consistency choice. The following examples show this distinction at
	the level of trajectories.

	\paragraph{Example 1: affine recurrence is algebraic.}
	Consider the affine network on two variables
	\[
	f(x_1,x_2)=\left(x_2,\ x_1\oplus x_2\right).
	\]
	For this network,
	\[
	f^2(x_1,x_2)=\left(x_1\oplus x_2,\ x_1\right),
	\qquad
	f^3(x_1,x_2)=\left(x_1,\ x_2\right).
	\]
	The configuration \(00\) is fixed, and the three non-zero configurations form
	the exact period-three attractor
	\[
	01\longmapsto 11\longmapsto 10\longmapsto 01.
	\]
	This example shows why the polynomial side of the dichotomy cannot be
	interpreted as a dynamical restriction to fixed points. Affine networks may
	have genuine non-stationary recurrence, but their recurrence equations remain
	linear: for an arbitrary affine network \(f(x)=Ax\oplus b\), the condition
	\(f^t(x)=x\) is a linear system over \(\mathbb F_2\) for every fixed \(t\).
	Exact period is obtained by removing solutions with smaller period, rather
	than by exploring the exponentially large state-transition graph.

	\paragraph{Example 2: one-sided recurrence is order-theoretic.}
	Now consider the conjunctive network
	\[
	f(x_1,x_2,x_3)=\left(x_2,\ x_1,\ x_1\wedge x_2\right).
	\]
	It has the exact period-two orbit
	\[
	010\longmapsto 100\longmapsto 010,
	\]
	and the fixed points \(000\) and \(111\). In terms of zero sets
	\(Z(x)=\left\{i:x_i=0\right\}\), the two-cycle is
	\[
	\{1,3\}\longmapsto \{2,3\}\longmapsto \{1,3\}.
	\]
	The first two coordinates exchange their values, while the third
	coordinate remains zero because a conjunction is zero as soon as one of its
	inputs is zero. Thus the orbit is non-stationary, but it follows a closure
	rule for zero sets rather than an arbitrary interaction among states. In a
	general conjunctive network, the condition \(f^t(x)=x\) can be
	read through how zeros propagate along the dependency graph over \(t\) steps.
	This is the order-theoretic reason why the conjunctive boundary leads to
	graph-reachability and strongly connected component conditions.

	The disjunctive boundary is the same phenomenon under coordinatewise
	complementation. The dual disjunctive network
	\[
	g(x_1,x_2,x_3)=\left(x_2,\ x_1,\ x_1\vee x_2\right)
	\]
	has the complementary exact period-two orbit
	\[
	101\longmapsto 011\longmapsto 101.
	\]
	In the dual view, sets of ones replace zero sets and satisfy the corresponding
	closure conditions. The two examples therefore represent the same tractable
	mechanism: cyclic behaviour is allowed, but the admissible cyclic patterns are
	controlled by monotone propagation in the dependency graph.

	\paragraph{Example 3: signal transmission can be conditional.}
	The hard side of the theorem relies on a different local mechanism. The
	ternary majority rule \(\operatorname{maj}\) in \(\MajClass\) can make signal
	transmission depend on two auxiliary inputs. With \(s\) viewed as a signal and
	\(u,v\) as auxiliary inputs,
	\[
	\operatorname{maj}(s,u,v)=
	\begin{cases}
	0, & u=v=0,\\
	s, & u\neq v,\\
	1, & u=v=1.
	\end{cases}
	\]
	The same local rule therefore has two distinct modes. If the auxiliary inputs
	are equal, the output is forced to a constant; if they are different, the
	signal passes through unchanged.

	This distinction can already be seen in a two-coordinate cycle module with
	auxiliary coordinates that remain fixed along the trajectory. For fixed
	auxiliary values \(u\) and \(v\), set
	\[
	h_{u,v}(c_1,c_2)=
	\left(\operatorname{maj}(c_2,u,v),c_1\right).
	\]
	When \(u\neq v\), the map on \((c_1,c_2)\) is
	\((c_1,c_2)\mapsto(c_2,c_1)\), so \(10\) and \(01\) form an exact
	period-two attractor. When \(u=v=0\) or \(u=v=1\), one coordinate is forced,
	and the non-stationary cycle disappears. Thus the existence of a periodic
	orbit can depend on whether an auxiliary layer can realise the balanced
	condition.

	The hardness reductions use this local mechanism at a larger scale. The
	periodic layer remains simple: it is a shift of a prescribed pattern. The
	nontrivial part is the auxiliary layer, whose local tests encode the
	consistency of a Boolean assignment. A cyclic attractor exists only when those
	tests allow the cycle signal to pass through. Thus the desired attractor
	remains an ordinary finite recurrent orbit, but its existence is tied to a
	global satisfiability constraint.

	\section{Discussion}

	We classify which local Boolean rules make prescribed-period cyclic-attractor
	detection tractable in the worst case. The update scheme is synchronous, the
	period is fixed in advance, and the interaction graph is not restricted. Under
	these assumptions, the boundary is determined by the closed class of
	admissible coordinate functions. Classes containing \(\MajClass\),
	\(\OrAndClass\), or \(\AndOrClass\) give \(\classNP\)-complete instances; all
	remaining closed classes are contained in \(\AffClass\), \(\ConjClass\), or
	\(\DisjClass\), where the problem is polynomial-time decidable. Thus the
	result separates the computational effects of different restricted local
	rules, rather than treating local-rule restrictions as a single form of
	regularity.

	The hard cases show that local regularity is not sufficient for efficient
	cyclic-attractor detection. The reductions use a fixed period and a simple
	cycle module; the source of hardness is the surrounding logical layer that
	decides whether the periodic signal can be transmitted. Majority-like
	self-dual rules and the two mixed monotone classes can express this
	conditional transmission while remaining highly restricted as Boolean
	functions. Consequently, fixing the period does not remove the possibility of
	encoding a global consistency test inside the network.

	The tractable cases have a different explanation. Affine rules remain affine
	under iteration, so the relevant periodicity equations reduce to linear
	systems over \(\mathbb F_2\). Purely conjunctive rules with constants can be
	analysed through the propagation of zero sets, and purely disjunctive rules
	through the dual propagation of sets of ones. These arguments do not rely on a
	small state-transition graph. They rely on algebraic or order-theoretic
	structure that is preserved by the network dynamics. This is why monotonicity
	alone does not mark the tractability boundary: some monotone classes are hard,
	while the one-sided conjunctive and disjunctive classes are tractable.

	This classification complements algorithms for concrete Boolean networks.
	Small-attractor search, SAT encodings, and symbolic methods avoid
	explicit traversal of the \(2^n\)-state transition graph
	\cite{zhang2007smallAttractors,dubrovaTeslenko2011,zheng2013attractors}.
	Other methods exploit additional structure: Answer Set Programming provides
	logical encodings for regulatory models \cite{khaled2023aspBooleanNetworks},
	trap spaces and related counting methods use invariant regions of Boolean
	dynamics \cite{klarner2015trapSpaces,kabir2025ScalableCounting}, and symbolic
	bifurcation analysis studies parameter-dependent changes in attractors
	\cite{benes2022AttractorBifurcations}. Network reductions and linear cuts
	provide reduction-based approaches to large models
	\cite{tonelloPauleve2023Reduction,naldi2023LinearCuts}. Attractor-landscape
	and stable-motif workflows give another practical route to model analysis
	\cite{trinh2025Biobalm,rozum2022Pystablemotifs}, while BooN provides a
	dedicated software environment for Boolean networks \cite{delaplace2025BooN}.
	Such methods may be effective because they exploit the particular model,
	update semantics, or application domain. The dichotomy proved here addresses a
	different worst-case question: whether the local rule class alone guarantees a
	polynomial-time decision procedure for every instance. For the tractable
	classes, the Methods provide such procedures; for the hard classes, such a
	guarantee would imply \(\classP=\classNP\).

	The difference from fixed-point detection has a dynamical explanation. A fixed
	point only has to satisfy a stationary equation, whereas a non-stationary cycle
	must also maintain a temporal pattern. The hard reductions use this additional
	temporal component by coupling a period-\(k\) signal to a satisfiability test.
	The tractable algorithms, in contrast, work because affine and one-sided
	closure structure remain stable under the iterates needed to test that
	temporal pattern.

	For Boolean-network modelling, the theorem is best interpreted as a statement
	about algebraic expressiveness, not as a claim about typical behaviour in
	applications. Boolean networks are widely used as simplified models of gene
	regulation \cite{kauffman,thomas,glassKauffman1973} and neural computation
	\cite{hopfield}, and modelling formalisms may impose additional restrictions
	on local rules. Canalization and threshold representations, for example,
	capture specific regulatory mechanisms and must be evaluated against the
	dynamics they are intended to represent
	\cite{kadelkaMurrugarra2024Canalization,kadelkaHari2025ThresholdModels}. Closed
	Boolean-function classes impose a different kind of restriction, based on
	closure properties of Boolean functions. The theorem should therefore not be
	read as a claim that application-derived networks are typically hard. It
	identifies which closed rule families are expressive enough to contain hard
	worst-case instances, and which affine or one-sided families exclude the
	encodings used in the hardness proofs. A related distinction arises in
	Boolean-network control. Controllability studies ask how interventions can
	steer networks toward selected attractors
	\cite{borrielloDaniels2021Controllability}, and influence maximization poses
	a related intervention problem on Boolean dynamics
	\cite{parmerRochaRadicchi2022Influence}. Work on difficult control and global
	stabilization studies how recurrent behaviour can be enforced or stabilized
	\cite{danielsBorriello2025DifficultControl,rafimanzelat2025GlobalStabilization}.
	Here, by contrast, we classify the complexity of detecting recurrent behaviour
	with restricted local rules.

	The classification applies to the recurrence problem studied here: an
	explicitly given finite Boolean map, synchronous update, and a fixed period.
	Changing the state space, the update semantics, or whether the period is fixed
	changes the mathematical problem. For
	switching systems, one may ask which attractors are retained after Boolean
	idealisation \cite{norrell2007attractors}. For continuous recurrent neural
	networks, recurrence is studied through fixed or slow points of trained
	systems \cite{sussilloBarak2013OpeningBlackBox}, symbolic encodings of
	piecewise-linear reconstructions \cite{brenner2024AlmostLinearRNN}, or fixed
	points and \(k\)-cycles in ReLU phase-space partitions
	\cite{eisenmann2023BifurcationsRNN}. Even within Boolean dynamics, changing
	the update semantics changes the object to be detected. Under asynchronous
	update, attractors are recurrent components rather than single periodic
	orbits, and signed-cycle structure or network reduction can play different
	roles \cite{richardTonello2023Separation,tonelloPauleve2023Reduction}.
	Sequential and block-sequential schedules lead to different cycle-complexity
	questions \cite{aracenaa}, while block-parallel update modes can make related
	dynamical decision problems complete for \(\classPSPACE\)
	\cite{perrotSeneTapin2024BlockParallel}. Natural extensions are therefore to
	treat the period as part of the input, combine local rule restrictions with
	graph restrictions, and compare the worst-case boundary with attractor
	statistics in random or curated Boolean-network models.

	\section{Methods}

	\subsection{Auxiliary notation and reductions}

	Throughout the Methods, \(k\ge 2\) is fixed. If \(z\in\mathbb B^n\) is a
	configuration and \(v\in[n]\), then \(z_v\) denotes the value of coordinate
	\(v\) in \(z\). For a trajectory \(z^{(0)},z^{(1)},\ldots\) of the global
	update map, \(z_v^{(t)}\) denotes the value of coordinate \(v\) at time \(t\).
	The same notation is used for original network coordinates and for auxiliary
	coordinates introduced in reductions. In running-time estimates, \(s\) denotes
	the total size of the input coordinate circuits, counting gates and wires.

	The hardness proofs use many-one reductions computable in time polynomial in
	the size of the source instance. For the boundary based on satisfiability, the
	source problem is the following standard decision problem.

	\begin{problem}[\(3\)-SAT]
		The input is a Boolean formula \(\psi\) in \(3\)-CNF on variables
		\(\{y_1,\ldots,y_n\}\). The question is whether there is an assignment
		\(a:\{y_1,\ldots,y_n\}\to\mathbb B\) such that \(\psi(a)=1\).
	\end{problem}

	An admissible instance over a closed class \(\mathcal C\) uses a fixed finite
	generating set \(\mathcal B_{\mathcal C}\), chosen before the input is given.
	Each coordinate function is represented by an acyclic circuit over
	\(\mathcal B_{\mathcal C}\) with inputs among \(x_1,\ldots,x_n\). Repeated
	inputs are allowed. Constants may be used only when they belong to
	\(\mathcal C\); if such a constant is not included as a primitive gate, a
	fixed \(\mathcal B_{\mathcal C}\)-circuit for it is substituted. Since the
	generating set is fixed and finite, gate fan-in is bounded by a constant. The
	instance size is \(n\) plus the total size of all coordinate circuits,
	counting gates and wires.

	The particular finite generating set does not affect the classification. If
	two fixed finite generating sets define the same closed class, each gate of
	one representation can be replaced by a fixed circuit over the other, giving
	polynomial-time equivalent encodings~\cite{lau2006}. When a proof writes a
	local rule using a fixed auxiliary function in the target class, the rule is
	still interpreted within this circuit model: fix one circuit for the
	auxiliary function over the target generating set and insert that circuit at
	each occurrence. These substitutions increase the output size by only a
	constant factor per occurrence.

	This representation convention is used throughout the proofs. Each reduction
	specifies the target local rule class, the size of the constructed network,
	and the equivalence between source instances and exact-period attractors. Each
	polynomial algorithm specifies the representation it uses and why, for fixed
	\(k\), the number of checks is polynomial.

	\subsection{Proof strategy}

	The dichotomy follows from boundary results and a lattice-theoretic lifting
	step. Table~\ref{tab:boundary-summary} summarises the boundary analysis: the
	hard lemmas give reductions from \(3\)-SAT or from an already established hard
	boundary class, whereas the tractable lemmas give algorithms based on affine
	linear algebra or one-sided graph reachability. Lemma~\ref{lem:monotonicity}
	then transfers results along inclusions of local rule classes, and
	Lemma~\ref{lem:post-separation} states that every closed class avoiding the
	three hard boundary classes is contained in one of the three tractable
	boundary classes. These two structural statements turn the six boundary
	results into the full dichotomy for all closed Boolean-function classes.

	\begin{table}[H]
		\centering
		\begingroup
		\small
		\setlength{\tabcolsep}{3pt}
		\begin{tabular}{p{0.16\textwidth}p{0.20\textwidth}p{0.42\textwidth}p{0.12\textwidth}}
			\hline
			Rule family & Complexity status & Proof mechanism & Lemma \\
			\hline
			\(\MajClass\) & \(\classNP\)-complete & Reduction from \(3\)-SAT using majority-gate conditional transmission & \ref{lem:d2} \\
			\(\AndOrClass\) & \(\classNP\)-complete & Control-coordinate reduction from \(\MajClass\) & \ref{lem:f6} \\
			\(\OrAndClass\) & \(\classNP\)-complete & Coordinatewise duality from \(\AndOrClass\) & \ref{lem:f2} \\
			\(\AffClass\) & Polynomial time & Affine lift, linear systems, and Möbius inversion & \ref{lem:l1} \\
			\(\ConjClass\) & Polynomial time & Zero-set reachability and strongly connected components & \ref{lem:p6-p} \\
			\(\DisjClass\) & Polynomial time & Coordinatewise duality from \(\ConjClass\) & \ref{lem:s6-p} \\
			\hline
		\end{tabular}
		\endgroup
		\caption{Proof map for the boundary classes used in the dichotomy for fixed
		\(k\ge2\). The hard rows are proved by polynomial-time reductions, while
		the tractable rows are proved by explicit algorithms.}
		\label{tab:boundary-summary}
	\end{table}

	\subsection{Hardness reductions}

	The first three boundary results prove \(\classNP\)-completeness. The
	starting point is the majority self-dual class, where conditional transmission
	encodes a \(3\)-SAT instance. The two mixed monotone hard classes then follow
	by reductions from this boundary case and by coordinatewise duality.

	\begin{lemma}[majority self-dual boundary]\label{lem:d2}
		For every fixed integer \(k\ge 2\),
		\(k\)-PLCE\((\MajClass)\) is \(\classNP\)-complete.
	\end{lemma}
	\begin{proof}
		Membership in \(\classNP\) follows because \(k\) is fixed: a certificate is
		a configuration \(x\), and one verifies \(f^k(x)=x\) and
		\(f^\ell(x)\neq x\) for all \(\ell<k\) by computing the first \(k\)
		iterates. For \(\classNP\)-hardness we reduce from \(3\)-SAT. Let
		\(\psi=\Gamma_1\wedge\cdots\wedge\Gamma_m\) be a \(3\)-CNF formula over
		variables \(y_1,\ldots,y_n\), with
		\(\Gamma_i=(L_{i1}\vee L_{i2}\vee L_{i3})\).
		We build a network \(f=\mathcal N_\psi\) in three layers. The literal layer has
		two fixed vertices \(v_j,\bar{v}_j\) for each variable \(y_j\), with local
		functions \(f_{v_j}(z)=z_{v_j}\) and
		\(f_{\bar{v}_j}(z)=z_{\bar{v}_j}\). The checking layer has clause vertices
		\(u_i\) and consistency vertices \(w_j\). The vertex \(u_i\) receives the
		three literal inputs corresponding to \(\Gamma_i\) and one self-loop. Its
		local rule is the ternary-majority circuit
		\[
		\theta(s,p,q,r)=
		\operatorname{maj}
		\bigl(
		\operatorname{maj}(s,p,q),
		\operatorname{maj}(s,p,r),
		\operatorname{maj}(s,q,r)
		\bigr),
		\]
		applied as
		\[
		f_{u_i}(z)=
		\theta(z_{u_i},z_{L_{i1}},z_{L_{i2}},z_{L_{i3}}).
		\]
		Here \(z_L\) means \(z_{v_j}\) when the literal \(L\) is
		\(y_j\), and \(z_{\bar{v}_j}\) when \(L\) is \(\neg y_j\).
		Thus, if \(s=1\), the rule returns
		\(p\vee q\vee r\); if \(s=0\), it returns \(p\wedge q\wedge r\).
		The vertex \(w_j\) receives \(v_j,\bar{v}_j\) and one self-loop and applies
		\[
		f_{w_j}(z)=
		\operatorname{maj}(z_{v_j},z_{\bar{v}_j},z_{w_j}).
		\]

		The cycle layer consists of vertices \(\delta_1,\ldots,\delta_k\). For
		\(r=2,\ldots,k\), set \(f_{\delta_r}(z)=z_{\delta_{r-1}}\). The vertex
		\(\delta_1\) is represented by an acyclic circuit over ternary-majority
		gates. Fix an ordering
		\[
		(i_1,j_1),\ldots,(i_{mn},j_{mn})
		\]
		of all pairs in \([m]\times[n]\). Put \(\chi_0=z_{\delta_k}\) and
		\[
		\chi_h=
		\operatorname{maj}\left(\chi_{h-1},z_{u_{i_h}},z_{w_{j_h}}\right),
		\qquad h=1,\ldots,mn.
		\]
		Then \(f_{\delta_1}(z)=\chi_{mn}\).
		Every described local function is a composition of ternary-majority gates
		and projections, hence belongs to \(\MajClass\). If the fixed basis for \(\MajClass\) is
		not the ternary-majority basis, each ternary-majority gate is replaced by its
		fixed circuit over that basis. The clause coordinates use four such gates,
		the consistency coordinates use one, and the \(\delta_1\)-coordinate uses
		\(mn\). Therefore the output instance has polynomial size.

		Figure~\ref{fig:d2-network-construction} summarises the role of the three
		layers in the reduction.

		\begin{figure}[H]
			\centering
			\resizebox{0.99\textwidth}{!}{%
			\begin{tikzpicture}[
				x=1cm,y=1cm,
				reductionLayer/.style={
					draw=csfInk!18,
					rounded corners=4pt,
					line width=0.45pt
				},
				reductionNode/.style={
					circle,
					draw=csfInk!60,
					fill=white,
					minimum size=7.5mm,
					inner sep=0pt,
					font=\small
				},
				reductionBox/.style={
					draw=csfInk!45,
					fill=white,
					rounded corners=2pt,
					line width=0.45pt,
					align=center,
					font=\small,
					inner sep=4pt
				},
				reductionArrow/.style={
					-{Latex[length=2.2mm]},
					line width=0.55pt,
					draw=csfInk!70
				},
				checkArrow/.style={reductionArrow, draw=csfBlue!75!black},
				cycleArrow/.style={reductionArrow, draw=csfTeal!75!black},
				feedArrow/.style={reductionArrow, draw=csfRed!72!black}
				]

				\fill[reductionLayer, fill=csfBlue!5] (-0.25,-2.05) rectangle (3.35,2.50);
				\fill[reductionLayer, fill=csfGray!30] (3.65,-2.05) rectangle (9.45,2.50);
				\fill[reductionLayer, fill=csfTeal!5] (9.80,-2.05) rectangle (14.85,2.50);

				\node[font=\bfseries\small, text=csfInk] at (1.55,2.15) {Literal layer};
				\node[font=\bfseries\small, text=csfInk] at (6.55,2.15) {Checking layer};
				\node[font=\bfseries\small, text=csfInk] at (12.33,2.15) {Cycle layer};
				\node[font=\scriptsize, text=csfInk!70] at (1.55,1.83)
					{fixed assignment vertices};
				\node[font=\scriptsize, text=csfInk!70] at (6.55,1.83)
					{\(\operatorname{maj}\)-gates};
				\node[font=\scriptsize, text=csfInk!70] at (12.33,1.83)
					{period-\(k\) shift register};

				\node[reductionNode, fill=csfBlue!9] (l1) at (0.75,0.82) {$v_1$};
				\node[reductionNode, fill=csfBlue!9] (lb1) at (2.35,0.82) {$\bar{v}_1$};
				\node[font=\large, text=csfInk!65] at (1.55,-0.10) {$\cdots$};
				\node[reductionNode, fill=csfBlue!9] (ln) at (0.75,-1.02) {$v_n$};
				\node[reductionNode, fill=csfBlue!9] (lbn) at (2.35,-1.02) {$\bar{v}_n$};
				\draw[checkArrow, draw=csfBlue!55!black] (l1) edge[loop above,
					looseness=5] (l1);
				\draw[checkArrow, draw=csfBlue!55!black] (lb1) edge[loop above,
					looseness=5] (lb1);
				\draw[checkArrow, draw=csfBlue!55!black] (ln) edge[loop below,
					looseness=5] (ln);
				\draw[checkArrow, draw=csfBlue!55!black] (lbn) edge[loop below,
					looseness=5] (lbn);

				\node[reductionBox, draw=csfBlue!55!black, fill=csfBlue!7,
				text width=4.70cm] (G) at (6.55,0.88)
					{\textbf{Clause tests}\\[-0.5mm]
					\(f_{u_i}=\theta(u_i,L_{i1},L_{i2},L_{i3})\)\\[-0.5mm]
					\(i=1,\ldots,m\)};
				\node[reductionBox, draw=csfGreen!60!black, fill=csfGreen!9,
				text width=4.70cm] (B) at (6.55,-0.88)
					{\textbf{Consistency tests}\\[-0.5mm]
					\(f_{w_j}=\operatorname{maj}(v_j,\bar{v}_j,w_j)\)\\[-0.5mm]
					\(j=1,\ldots,n\)};

				\draw[checkArrow] (2.90,0.82) -- (G.west);
				\draw[checkArrow, draw=csfGreen!70!black] (2.90,-0.82) -- (B.west);

				\node[reductionNode, draw=csfRed!70!black, fill=csfRed!9] (d1)
					at (11.05,0.00) {$\delta_1$};
				\node[reductionNode, draw=csfTeal!75!black, fill=csfTeal!10] (d2)
					at (12.25,1.20) {$\delta_2$};
				\node[font=\large, text=csfInk!65] (ddots) at (13.45,0.00) {$\cdots$};
				\node[reductionNode, draw=csfTeal!75!black, fill=csfTeal!10] (dk)
					at (12.25,-1.20) {$\delta_k$};

				\draw[cycleArrow] (11.14,0.48)
					arc[start angle=157,end angle=113,radius=1.22];
				\draw[cycleArrow] (12.72,1.14)
					arc[start angle=67,end angle=23,radius=1.22];
				\draw[cycleArrow] (13.36,-0.48)
					arc[start angle=-23,end angle=-67,radius=1.22];
				\draw[cycleArrow] (11.78,-1.14)
					arc[start angle=-113,end angle=-157,radius=1.22];

				\draw[feedArrow] (G.east) .. controls (9.55,0.88) and (10.05,0.40) ..
					node[above, pos=0.38, font=\scriptsize, text=csfRed!75!black,
					fill=white, inner sep=0.6pt] {all pairs}
					(d1.160);
				\draw[feedArrow] (B.east) .. controls (9.55,-0.88) and (10.05,-0.40) ..
					node[below, pos=0.38, font=\scriptsize, text=csfRed!75!black,
					fill=white, inner sep=0.6pt] {all pairs}
					(d1.-160);
			\end{tikzpicture}}
			\caption{Layered construction of the \(\MajClass\)-network
			\(\mathcal N_\psi\) used in Lemma~\ref{lem:d2}. The literal vertices are
			fixed by self-loops. The checking vertices encode clause satisfaction and
			consistency by ternary-majority circuits. The cycle layer is a circular shift
			register; \(\delta_1\) is represented by a ternary-majority chain over all
			pairs \((u_i,w_j)\). This chain transmits the value of \(\delta_k\) exactly
			when every clause-test value is opposite to every consistency-test value.}
			\label{fig:d2-network-construction}
		\end{figure}
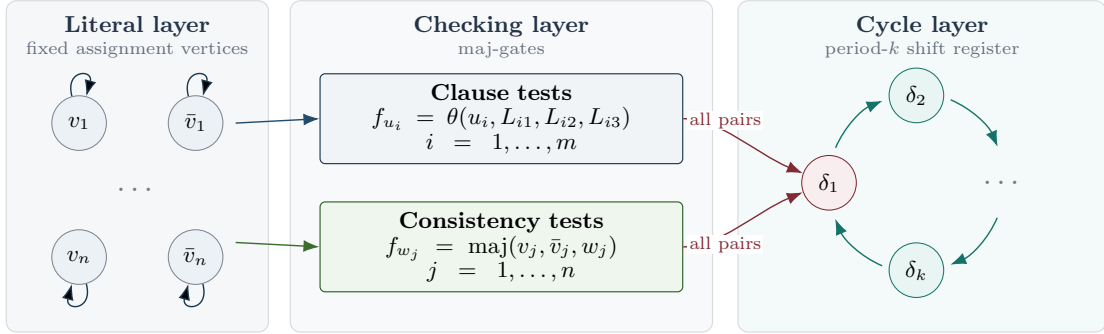

		Assume first that \(\psi\) is satisfied by an assignment \(a\). Put
		\(z_{v_j}^{(0)}=a(y_j)\), \(z_{\bar{v}_j}^{(0)}=1-a(y_j)\),
		\(z_{u_i}^{(0)}=1\), \(z_{w_j}^{(0)}=0\), and
		\(\left(z_{\delta_1}^{(0)},\ldots,z_{\delta_k}^{(0)}\right)
		=\left(1,0,\ldots,0\right)\). The literal values are fixed. Since each
		clause contains a true literal and \(\theta(1,p,q,r)=p\vee q\vee r\), every
		\(u_i\) remains equal to \(1\). Each \(w_j\) sees exactly one of
		\(v_j,\bar{v}_j\) equal to \(1\) and its own value equal to \(0\), so
		it remains equal to \(0\). Hence the checking layer is fixed. At
		\(\delta_1\), every pair \((u_i,w_j)\) has unequal values. Each gate in the
		chain therefore transmits its first input, so the update at \(\delta_1\) is
		the value of \(\delta_k\).
		Together with \(f_{\delta_r}=z_{\delta_{r-1}}\) for \(r\ge 2\), the cycle
		layer performs a cyclic shift of \(\left(1,0,\ldots,0\right)\). Therefore
		the constructed network has a cyclic attractor of exact period \(k\).

		Conversely, suppose that \(z^{(0)}\in\Phi_k(f)\) for the constructed network.
		The literal vertices are constant. Once these literal values are fixed, each
		checking vertex has only its own current value as a time-dependent input:
		\(u_i\) evolves by a map
		\(\tau_i(s)=\theta(s,z_{L_{i1}},z_{L_{i2}},z_{L_{i3}})\), and
		\(w_j\) evolves by a map
		\(\sigma_j(s)=\operatorname{maj}(z_{v_j},z_{\bar{v}_j},s)\). Both maps are
		monotone maps from \(\mathbb B\) to itself. Since the whole trajectory is
		periodic, the value sequence of each checking vertex is periodic under its
		one-dimensional map. A nonconstant periodic sequence on \(\mathbb B\) would
		have to alternate \(0,1,0,1,\ldots\), which would force
		the corresponding one-dimensional map to send \(0\) to \(1\) and \(1\) to
		\(0\). This contradicts monotonicity. Hence all checking vertices are
		constant on the periodic trajectory.
		Let \(\gamma_i\) and \(\beta_j\) denote the constant values of \(u_i\) and
		\(w_j\), respectively, on this trajectory.

		If some pair \((\gamma_i,\beta_j)\) has equal values, consider the last such
		pair in the fixed ordering used by the \(\delta_1\)-circuit. At the corresponding
		ternary-majority gate, the output becomes that common value, independently of
		the incoming value. All later pairs are unequal and therefore only transmit
		the current value. Hence the final output of the \(\delta_1\)-circuit is
		independent of \(\delta_k\). Then, after at most \(k\) updates, all vertices
		\(\delta_1,\ldots,\delta_k\) are constant and remain constant thereafter.
		Periodicity would force them to have been constant already, contradicting
		exact period \(k\ge 2\). Therefore \(\gamma_i\neq \beta_j\) for every \(i\)
		and \(j\).

		The last condition implies that all \(\gamma_i\) have the same value and all
		\(\beta_j\) have the opposite value. Thus either \(\gamma_i=1\) for all \(i\)
		and \(\beta_j=0\) for all \(j\), or \(\gamma_i=0\) for all \(i\) and
		\(\beta_j=1\) for all \(j\). The consistency gate gives the only truth-table
		facts needed in the two cases that follow: if \(\beta_j=0\), then \(z_{v_j}\) and
		\(z_{\bar{v}_j}\) cannot both be \(1\), because
		\(\operatorname{maj}(1,1,0)=1\); if \(\beta_j=1\), then they cannot both be
		\(0\), because \(\operatorname{maj}(0,0,1)=0\).

		In the first case, define \(a(y_j)=z_{v_j}\) for every variable \(y_j\).
		Fix a clause \(\Gamma_i\). Since \(u_i\) is constant with value \(1\), its
		fixed-value equation gives
		\[
		1=\theta(1,z_{L_{i1}},z_{L_{i2}},z_{L_{i3}})
		=
		z_{L_{i1}}\vee z_{L_{i2}}\vee z_{L_{i3}} .
		\]
		Thus \(z_{L_{it}}=1\) for some \(t\in\{1,2,3\}\). If
		\(L_{it}=y_j\), then \(a(y_j)=z_{v_j}=1\). If
		\(L_{it}=\neg y_j\), then \(z_{\bar{v}_j}=1\), and the equality
		\(0=\operatorname{maj}(z_{v_j},z_{\bar{v}_j},0)\) forces
		\(z_{v_j}=0\). Hence \(a(y_j)=0\), so the literal \(\neg y_j\) is true.
		In either subcase, \(\Gamma_i\) is satisfied by \(a\).

		In the second case, define \(a(y_j)=1-z_{v_j}\) for every variable
		\(y_j\). Fix a clause \(\Gamma_i\). Since \(u_i\) is constant with value
		\(0\), its fixed-value equation gives
		\[
		0=\theta(0,z_{L_{i1}},z_{L_{i2}},z_{L_{i3}})
		=
		z_{L_{i1}}\wedge z_{L_{i2}}\wedge z_{L_{i3}} .
		\]
		Thus \(z_{L_{it}}=0\) for some \(t\in\{1,2,3\}\). If
		\(L_{it}=y_j\), then \(z_{v_j}=0\) and \(a(y_j)=1\). If
		\(L_{it}=\neg y_j\), then \(z_{\bar{v}_j}=0\), and the equality
		\(1=\operatorname{maj}(z_{v_j},z_{\bar{v}_j},1)\) forces
		\(z_{v_j}=1\). Hence \(a(y_j)=0\), so the literal \(\neg y_j\) is true.
		Again, \(\Gamma_i\) is satisfied by \(a\). In both cases every clause of
		\(\psi\) is satisfied. The reduction is correct, so
		\(k\)-PLCE\((\MajClass)\) is \(\classNP\)-hard; together with membership in
		\(\classNP\), it is \(\classNP\)-complete.
	\end{proof}

	\begin{lemma}[mixed monotone AND-over-OR boundary]\label{lem:f6}
		For every fixed integer \(k\ge 2\),
		\(k\)-PLCE\((\AndOrClass)\) is \(\classNP\)-complete.
	\end{lemma}
	\begin{proof}
		Membership in \(\classNP\) follows because \(k\) is fixed and a candidate configuration
		can be verified by computing its first \(k\) iterates. For hardness we reduce
		from \(k\)-PLCE\((\MajClass)\). Given a \(\MajClass\)-network
		\(f=(f_1,\ldots,f_n)\), add a control coordinate \(y\) and define a network
		\(g=(g_0,g_1,\ldots,g_n)\) by
		\[
		g_0(y,x)=y,\qquad g_i(y,x)=y\wedge f_i(x),\quad i\in[n].
		\]
		The projection \(g_0\) belongs to \(\AndOrClass\). For \(i\ge1\), the
		function \(f_i\in \MajClass\) is monotone and self-dual, so
		\(f_i(0,\ldots,0)=0\) and \(f_i(1,\ldots,1)=1\). Hence
		\(g_i=y\wedge f_i\) is monotone, preserves \(0\) and \(1\), and is bounded
		above by \(y\). By the standard description of \(\AndOrClass\)
		\cite{lau2006}, each coordinate lies in \(\AndOrClass\).

		It remains to make the circuit representation explicit. Let
		\(\mathcal B_{\MajClass}\) and \(\mathcal B_{\AndOrClass}\) be the fixed bases from
		the input encoding. For each gate \(b\in\mathcal B_{\MajClass}\) of arity \(r\),
		the function
		\[
		\widehat b(y,u_1,\ldots,u_r)=y\wedge b(u_1,\ldots,u_r)
		\]
		belongs to \(\AndOrClass\), and so does
		\(\alpha(y,u)=y\wedge u\). We fix
		\(\mathcal B_{\AndOrClass}\)-circuits for these finitely many functions.
		Given a \(\mathcal B_{\MajClass}\)-circuit for \(f_i\), replace each
		\(b\)-gate by the fixed circuit for \(\widehat b\) and put the fixed
		\(\alpha\)-circuit at the output. A direct induction over the source circuit
		shows that the resulting circuit computes \(y\wedge f_i(x)\). Its size is
		linear in the size of the source circuit, so the output instance is
		admissible and has polynomial size.

		If \(x^{(0)}\in\Phi_k(f)\), then
		\(z^{(0)}=\left(1,x^{(0)}\right)\) lies in \(\Phi_k(g)\): the control
		coordinate remains equal to \(1\), and the remaining coordinates follow
		exactly the trajectory of \(f\). Conversely, let
		\(z^{(0)}=\left(y^{(0)},x^{(0)}\right)\in\Phi_k(g)\). Since \(g_0(y,x)=y\),
		the control coordinate is constant. If \(y^{(0)}=0\), all original
		coordinates become \(0\) after one step, so no cycle of exact period
		\(k\ge 2\) is possible. Hence \(y^{(0)}=1\), and then the \(x\)-coordinates
		satisfy \(x^{(t+1)}=f\left(x^{(t)}\right)\). Exact period \(k\) for
		\(z^{(0)}\) therefore implies exact period \(k\) for \(x^{(0)}\) under \(f\).

		The construction is polynomial-time and preserves positive instances in both
		directions.
		Since \(k\)-PLCE\((\MajClass)\) is \(\classNP\)-complete by
		Lemma~\ref{lem:d2}, \(k\)-PLCE\((\AndOrClass)\) is \(\classNP\)-complete.
	\end{proof}

	\begin{lemma}[mixed monotone OR-over-AND boundary]\label{lem:f2}
		For every fixed integer \(k\ge 2\),
		\(k\)-PLCE\((\OrAndClass)\) is \(\classNP\)-complete.
	\end{lemma}
	\begin{proof}
		Membership in \(\classNP\) follows from the same fixed-period verification.
		The same lattice description identifies \(\OrAndClass\) as the
		coordinatewise dual of \(\AndOrClass\): if
		\(h\in \AndOrClass\), then
		\(h^\ast(x)=\neg h(\neg x_1,\ldots,\neg x_n)\) belongs to
		\(\OrAndClass\), and conversely~\cite{lau2006}. Given an
		\(\AndOrClass\)-network \(f\), define the coordinatewise dual network
		\[
		g_i(x)=\neg f_i(\neg x_1,\ldots,\neg x_n),\qquad i\in[n].
		\]
		Then \(g\) is an \(\OrAndClass\)-network.

		The transformation also respects the fixed-basis representation. Let
		\(\mathcal B_{\AndOrClass}\) and \(\mathcal B_{\OrAndClass}\) be the selected
		bases. For each \(b\in\mathcal B_{\AndOrClass}\), its dual
		\(b^\ast(u)=\neg b(\neg u)\) belongs to \(\OrAndClass\), so we fix a
		\(\mathcal B_{\OrAndClass}\)-circuit for \(b^\ast\). Replacing every
		\(b\)-gate in the input circuit by this fixed dual circuit gives, by
		bottom-up induction, a circuit for \(g_i\). The transformation is
		polynomial-time.

		The complement map \(c(x)=\neg x\) satisfies \(g(c(x))=c(f(x))\). Hence
		\(x\in\Phi_k(f)\) if and only if \(c(x)\in\Phi_k(g)\). The polynomial-time
		transformation and Lemma~\ref{lem:f6} show that
		\(k\)-PLCE\((\OrAndClass)\) is \(\classNP\)-complete.
	\end{proof}

	\subsection{Polynomial-time algorithms}

	The remaining boundary classes are tractable. The affine case reduces exact
	periodicity to linear algebra, while the conjunctive and disjunctive cases
	reduce it to reachability conditions for zero sets or, dually, one sets.

	\begin{lemma}[affine boundary]\label{lem:l1}
		For every fixed integer \(k\ge 2\),
		\(k\)-PLCE\((\AffClass)\in \classP\).
	\end{lemma}
	\begin{proof}
		Let \(f=(f_1,\ldots,f_n)\) be an \(\AffClass\)-network. Its global map is affine
		over \(\mathbb F_2\). From the input circuits, the affine representation of
		each coordinate is computed bottom-up: each gate contributes the xor of the
		coefficient vectors of its inputs, and the constant gate contributes the
		constant term. Repeated input wires are handled by adding the same coefficient
		vector again over \(\mathbb F_2\). With the notation above, this bottom-up
		pass takes \(O(sn)\) time, because each gate updates a length-\(n\)
		coefficient vector. Thus we obtain
		\(f(x)=Ax\oplus b\), with
		\(A\in\mathbb F_2^{n\times n}\) and \(b\in\mathbb F_2^n\). For \(t\ge 1\), let
		\(\operatorname{Fix}(t)=
		\left|\left\{x\in\mathbb F_2^n:f^t(x)=x\right\}\right|\), and let
		\(N(d)\) be the number of configurations of minimal period \(d\). Since a
		point is fixed by \(f^t\)
		exactly when its minimal period divides \(t\),
		\(\operatorname{Fix}(t)=\sum_{d\mid t}N(d)\). Möbius inversion gives
		\[
		N(k)=\sum_{d\mid k}\mu\left(k/d\right)\operatorname{Fix}(d),
		\]
		where \(\mu\) is the Möbius function. Thus the problem reduces to computing
		\(\operatorname{Fix}(d)\) for all divisors \(d\mid k\) and testing whether
		\(N(k)>0\).

		Introduce the standard affine lift
		\[
		M=
		\begin{pmatrix}
			A & b\\
			0 & 1
		\end{pmatrix}
		\in \mathbb F_2^{(n+1)\times(n+1)}.
		\]
		By induction, \(M^t\) has the form
		\[
		M^t=
		\begin{pmatrix}
			A_t & b_t\\
			0 & 1
		\end{pmatrix},
		\qquad\text{with}\qquad
		f^t(x)=A_t x\oplus b_t.
		\]
		For each divisor \(d\mid k\), the equation \(f^d(x)=x\) is therefore the
		linear system \((A_d\oplus I)x=b_d\) over \(\mathbb F_2\). Gaussian
		elimination gives the number \(\operatorname{Fix}(d)\) of its solutions in
		\(O\left(n^3\right)\) time: if the system is consistent and has rank \(r\),
		then \(\operatorname{Fix}(d)=2^{n-r}\), and otherwise
		\(\operatorname{Fix}(d)=0\). The matrices
		\(M,\ldots,M^k\) are obtained
		sequentially in \(O\left(kn^3\right)\) time, and Gaussian elimination for
		the relevant divisors contributes another \(O\left(kn^3\right)\). Computing
		\(N(k)\) from the Möbius formula then adds only \(O\left(k\right)\)
		arithmetic operations. Hence existence of a cyclic attractor of exact period
		\(k\) is decidable in \(O\left(sn+kn^3\right)\) time.
		Since \(k\) is fixed in the decision problem, the total running time is
		polynomial in the input size, and the problem belongs to \(\classP\).
	\end{proof}

	\begin{lemma}[conjunctive boundary]\label{lem:p6-p}
		For every fixed integer \(k\ge 2\),
		\(k\)-PLCE\((\ConjClass)\in \classP\).
	\end{lemma}
	\begin{proof}
		Recall that every \(\ConjClass\)-function is either \(0\), \(1\), or a conjunction
		of variables. Since the basis is fixed, each coordinate circuit is
		normalised to this form in polynomial time by a bottom-up pass: each basis
		gate has a precomputed \(0\), \(1\), or conjunction-of-selected-arguments
		description, and applying it to already normalised inputs only propagates
		constants and takes unions of input sets. If the \(v\)-th coordinate is not
		identically \(0\), write \(\operatorname{In}(v)\) for its input set and
		\[
		f_v(x)=\bigwedge_{u\in \operatorname{In}(v)}x_u,
		\]
		with the empty conjunction interpreted as \(1\). We compute the same
		representation for \(f_v^t\), \(1\le t\le k\): either \(f_v^t\equiv0\), or
		\(f_v^t(x)=\bigwedge_{u\in \operatorname{In}_t(v)}x_u\). The induction step
		substitutes the already computed representations of the coordinates in
		\(\operatorname{In}(v)\). If one substituted coordinate is identically
		\(0\), then \(f_v^{t+1}\equiv0\); otherwise
		\[
		\operatorname{In}_{t+1}(v)=
		\bigcup_{w\in \operatorname{In}(v)}\operatorname{In}_t(w).
		\]
		With input sets stored as \(n\)-bit vectors, the initial normalisation costs
		\(O(sn)\), and the computation of \(f,\ldots,f^k\) costs
		\(O\left(kn^3\right)\), since for each \(t\le k\) and each coordinate at
		most \(n\) length-\(n\) sets are unioned. Let \(H_t=([n],E_t)\) be
		the dependency graph of
		\(f^t\), with \((u,v)\in E_t\) iff
		\(u\in \operatorname{In}_t(v)\); if \(f_v^t\equiv0\), no edge is put into
		\(v\).

		For conjunctive maps, fixed-point conditions are naturally expressed through
		the zero set \(Z=\{v:x_v=0\}\). Let
		\(Z_0^k=\{v:f_v^k\equiv0\}\) be the coordinates forced to zero by the
		\(k\)-th iterate. Then \(f^k(x)=x\) if and only if
		\[
		Z_0^k\subseteq Z,\qquad
		Z \text{ is successor-closed in } H_k,\qquad
		\operatorname{In}_k(v)\cap Z\neq\emptyset
		\text{ for every }v\in Z\setminus Z_0^k .
		\]
		Here successor-closed means that every out-neighbour in \(H_k\) of a vertex
		in \(Z\) also belongs to \(Z\). Indeed, a coordinate equal to \(1\) remains
		\(1\) exactly when none of its inputs in \(H_k\) is zero, while a nonconstant
		coordinate equal to \(0\) remains \(0\) exactly when at least one input is
		zero.

		This criterion characterises configurations fixed by \(f^k\). To obtain
		exact period \(k\), we must also exclude every smaller period: for each
		\(\ell=1,\ldots,k-1\), at least one coordinate must satisfy
		\(f_v^\ell(x)\neq x_v\). To record a possible coordinate mismatch, we use a
		pair \((R,W)\) of vertex sets. A zero set \(Z\) realises the pair
		\((R,W)\) if \(R\subseteq Z\) and \(W\cap Z=\emptyset\); equivalently, the
		vertices in \(R\) are assigned value \(0\) and the vertices in \(W\) are
		assigned value \(1\). For a fixed \(\ell\), define
		\(A_\ell\) to be the list of all local witnesses for such a disagreement,
		over all coordinates \(v\):
		\[
		(\emptyset,\{v\})
		\quad\text{if } f_v^\ell\equiv0,
		\]
		\[
		(\{u\},\{v\})
		\quad\text{for each }u\in \operatorname{In}_\ell(v)
		\quad\text{if } f_v^\ell\not\equiv0,
		\]
		and
		\[
		(\{v\},\operatorname{In}_\ell(v))
		\quad\text{if } f_v^\ell\not\equiv0.
		\]
		These three types exactly describe the ways in which the \(v\)-th
		coordinate of \(f^\ell(x)\) can differ from \(x_v\). Let \(Z=\{u:x_u=0\}\).
		If \(f_v^\ell\equiv0\), then \(f_v^\ell(x)\neq x_v\) holds exactly when
		\(x_v=1\), that is, when \(v\notin Z\); this is the witness
		\((\emptyset,\{v\})\). If
		\(f_v^\ell(x)=\bigwedge_{u\in\operatorname{In}_\ell(v)}x_u\), there are two
		possibilities. First, \(x_v=1\) but the conjunction is \(0\); this is
		equivalent to \(v\notin Z\) and \(u\in Z\) for at least one
		\(u\in\operatorname{In}_\ell(v)\), giving a witness
		\((\{u\},\{v\})\). Second, \(x_v=0\) but the conjunction is \(1\); this is
		equivalent to \(v\in Z\) and
		\(\operatorname{In}_\ell(v)\cap Z=\emptyset\), giving the witness
		\((\{v\},\operatorname{In}_\ell(v))\). Conversely, realising any one of
		these witnesses forces the corresponding coordinate to disagree. Hence, for
		the configuration with zero set \(Z\), realising some witness in \(A_\ell\)
		is equivalent to \(f^\ell(x)\neq x\).

		As a small example, if \(f_v^\ell=x_a\wedge x_b\), then the witnesses
		attached to coordinate \(v\) are
		\((\{a\},\{v\})\), \((\{b\},\{v\})\), and \((\{v\},\{a,b\})\). They represent,
		respectively, \(x_a=0,x_v=1\), \(x_b=0,x_v=1\), and
		\(x_v=0,x_a=x_b=1\), which are exactly the three ways in which
		\(x_a\wedge x_b\) can differ from \(x_v\). Each coordinate contributes at most
		\(\left|\operatorname{In}_\ell(v)\right|+2\le n+2\) witnesses, so
		\(|A_\ell|=O\left(n^2\right)\).

		It remains to combine the inequalities \(f^\ell(x)\neq x\), \(\ell<k\),
		with the fixed-point criterion for \(f^k\). The algorithm considers all
		choices of witnesses, one pair \((R_\ell,W_\ell)\in A_\ell\) for each
		\(\ell=1,\ldots,k-1\). For such a choice, put
		\[
		R=Z_0^k\cup\bigcup_{\ell=1}^{k-1}R_\ell,
		\qquad
		W=\bigcup_{\ell=1}^{k-1}W_\ell .
		\]
		The set \(R\) collects all coordinates required to be zero: the coordinates
		forced to zero by \(f^k\), together with the zero requirements from the
		selected witnesses. The set \(W\) collects all coordinates required to be
		one by these witnesses. Call this witness choice feasible if
		\(R\cap W=\emptyset\) and there exists a zero set \(Z\) such that
		\(R\subseteq Z\), \(Z\cap W=\emptyset\), and \(Z\) satisfies the
		fixed-point criterion. Such a \(Z\) realises all selected witnesses
		simultaneously and satisfies \(f^k(x)=x\); therefore it gives
		\(f^\ell(x)\neq x\) for every \(\ell<k\). Conversely, if
		\(x\in\Phi_k(f)\), then for each \(\ell<k\) some coordinate of
		\(f^\ell(x)\) disagrees with \(x\), so the exactness condition selects one
		witness from \(A_\ell\) that is realised by the zero set of \(x\).

		It remains to test feasibility under the assumption \(R\cap W=\emptyset\).
		Define
		\[
		U=\left\{v\in[n]\;\middle|\;
		\text{no vertex of }W\text{ is reachable from }v\text{ in }H_k\right\}.
		\]
		Reachability includes paths of length \(0\). Then \(U\) is the largest
		successor-closed set disjoint from \(W\). Hence every zero set satisfying the
		feasibility conditions is contained in \(U\); if \(R\not\subseteq U\), no such
		zero set exists. Assume from now on that \(R\subseteq U\).

		Write \(H_k[U]\) for the subgraph of \(H_k\) induced by \(U\).
		In \(H_k[U]\), call a strongly connected component cyclic if it has more than
		one vertex or has a loop. Let \(Y_0\) be the union of \(Z_0^k\) and all cyclic
		strongly connected components of \(H_k[U]\), and let \(Y\) be the set of
		vertices reachable from \(Y_0\) inside \(H_k[U]\). We claim that the witness
		choice is feasible if and only if \(R\subseteq Y\).

		For necessity, let \(Z\subseteq U\) satisfy the feasibility conditions.
		Vertices in \(Z\cap Z_0^k\) already lie in \(Y_0\subseteq Y\). Now take
		\(v\in Z\setminus Z_0^k\). By the predecessor condition in the fixed-point
		criterion, there is an edge \(u\to v\) in \(H_k\) with \(u\in Z\). Repeating
		this step inside the finite set \(Z\) either reaches \(Z_0^k\) or repeats a
		vertex. In the first case, the resulting path runs forward from
		\(Z_0^k\subseteq Y_0\) to \(v\). In the second, the repeated segment is a
		directed cycle in \(H_k[U]\), hence lies in a cyclic component included in
		\(Y_0\), and the remaining path again runs forward to \(v\). Thus every zero
		set satisfying the feasibility conditions is contained in \(Y\), so
		\(R\subseteq Y\) is necessary.

		Conversely, suppose that \(R\subseteq Y\). We show that \(Y\) can be used as
		the required zero set. Since \(Y\subseteq U\), no vertex of \(W\) lies in
		\(Y\), and the selected witnesses are realised because \(R\subseteq Y\).
		Also \(Z_0^k\subseteq Y_0\subseteq Y\). The set \(Y\) is successor-closed in
		\(H_k\): a successor in \(H_k\) of a vertex of \(Y\) cannot leave \(U\), and
		inside \(H_k[U]\) it is again reachable from \(Y_0\). Finally, every vertex
		\(v\in Y\setminus Z_0^k\) has an incoming edge from a vertex of \(Y\). If
		\(v\) lies in a cyclic component included in \(Y_0\), take the
		predecessor inside that component; otherwise take the previous vertex on a
		path in \(H_k[U]\) from \(Y_0\) to \(v\). Thus \(Y\) satisfies the
		fixed-point criterion and realises the selected witnesses.

		Thus \(R\subseteq Y\) decides feasibility of the selected witnesses. This
		gives the following decision procedure.
		\begin{enumerate}
			\item Normalise every coordinate circuit to \(0\), \(1\), or a conjunction
			of variables. For \(t=1,\ldots,k\), compute the same representation for
			\(f^t\). From these representations construct \(H_k\), \(Z_0^k\), and
			the witness lists \(A_1,\ldots,A_{k-1}\).
			\item For each tuple
			\[
			\bigl((R_1,W_1),\ldots,(R_{k-1},W_{k-1})\bigr)
			\in A_1\times\cdots\times A_{k-1},
			\]
			set
			\(R=Z_0^k\cup\bigcup_{\ell=1}^{k-1}R_\ell\) and
			\(W=\bigcup_{\ell=1}^{k-1}W_\ell\). If \(R\cap W\neq\emptyset\), continue
			to the next tuple. Compute the largest successor-closed set \(U\) in
			\(H_k\) that is disjoint from \(W\). If \(R\not\subseteq U\), continue to
			the next tuple. Inside \(H_k[U]\), let \(Y_0\) be the union of \(Z_0^k\)
			and all cyclic strongly connected components, and let \(Y\) be the set
			of vertices reachable from \(Y_0\). If \(R\subseteq Y\), accept;
			otherwise continue to the next tuple.
			\item If the loop over witness tuples finishes without accepting, reject.
		\end{enumerate}
		By this equivalence, the procedure accepts if and only if
		\(\Phi_k(f)\neq\emptyset\). There are at most
		\(O\left(n^{2(k-1)}\right)\) witness choices. Constructing \(H_k\),
		\(Z_0^k\), and the lists \(A_1,\ldots,A_{k-1}\) from the computed
		input-set representations costs \(O\left(kn^2\right)\), which is dominated
		by the \(O\left(kn^3\right)\) computation of the iterates. For each witness
		choice, \(U\) is computed by one reverse reachability search in \(H_k\), and
		\(Y\) is computed by strongly connected components in \(H_k[U]\) followed
		by reachability from \(Y_0\). Since \(H_k\) has \(n\) vertices and at most
		\(n^2\) edges, one feasibility test takes \(O\left(n^2\right)\) time, and
		the whole witness loop takes \(O\left(n^{2k}\right)\) time. Combining this
		with the \(O(sn)\) normalisation of the input circuits and the
		\(O\left(kn^3\right)\) computation of \(f,\ldots,f^k\), the running time is
		\(O\left(sn+kn^3+n^{2k}\right)\).
		Therefore, for fixed \(k\),
		\(k\)-PLCE\((\ConjClass)\in \classP\).
	\end{proof}

	\begin{lemma}[disjunctive boundary]\label{lem:s6-p}
		For every fixed integer \(k\ge 2\),
		\(k\)-PLCE\((\DisjClass)\in \classP\).
	\end{lemma}
	\begin{proof}
		Recall that \(\ConjClass=[\{x\wedge y,0,1\}]\) and
		\(\DisjClass=[\{x\vee y,0,1\}]\). Given an \(\DisjClass\)-network
		\(f=(f_1,\ldots,f_n)\), define the coordinatewise dual network
		\[
		g_i(x_1,\ldots,x_n)
		=
		\neg f_i(\neg x_1,\ldots,\neg x_n),
		\qquad i\in[n].
		\]
		By De Morgan's law, disjunctions become conjunctions and the constants \(0\)
		and \(1\) are interchanged, so \(g\) is a \(\ConjClass\)-network and is computable
		from \(f\) in polynomial time. In the fixed-basis representation, each
		\(\DisjClass\)-basis gate is replaced by the precomputed \(\ConjClass\)-circuit for its
		coordinatewise dual, so the output instance is again admissible.

		Let \(c(x_1,\ldots,x_n)=(\neg x_1,\ldots,\neg x_n)\). By construction,
		\(g(c(x))=c(f(x))\), so \(g\circ c=c\circ f\). Induction gives
		\(g^t(c(x))=c(f^t(x))\) for every \(t\ge1\). Hence \(x\in\Phi_k(f)\) if and
		only if \(c(x)\in\Phi_k(g)\), and therefore
		\(\Phi_k(f)\neq\emptyset\) iff \(\Phi_k(g)\neq\emptyset\). The reduction to
		\(k\)-PLCE\((\ConjClass)\) has linear-size dualisation, which is dominated by the
		\(O\left(sn+kn^3+n^{2k}\right)\) bound of the \(\ConjClass\) algorithm; hence
		Lemma~\ref{lem:p6-p} completes the proof.
	\end{proof}

	Table~\ref{tab:tractable-complexity} summarises the running-time bounds for
	the tractable boundary algorithms. Here \(s\) denotes the total size of the
	input coordinate circuits. Since \(k\) is fixed, all three bounds are
	polynomial in the input size.

	\begin{table}[H]
		\centering
		\begingroup
		\small
		\setlength{\tabcolsep}{4pt}
		\begin{tabular}{>{\raggedright\arraybackslash}p{0.14\textwidth}
			>{\raggedright\arraybackslash}p{0.42\textwidth}
			>{\raggedright\arraybackslash}p{0.28\textwidth}}
			\hline
			Class & Main operations & Running time \\
			\hline
			\(\AffClass\) &
			\(O(sn)\) conversion to an affine map; \(O(kn^3)\) matrix powers and
			Gaussian eliminations. &
			\(O\left(sn+kn^3\right)\) \\
			\(\ConjClass\) &
			\(O(sn)\) normalisation to conjunctions; \(O(kn^3)\) computation of
			\(f,\ldots,f^k\); \(O\left(n^{2(k-1)}\right)\) witness tuples, each
			tested in \(O\left(n^2\right)\) time. &
			\(O\left(sn+kn^3+n^{2k}\right)\) \\
			\(\DisjClass\) &
			Linear-size coordinatewise dualisation followed by the \(\ConjClass\)
			preprocessing and witness test. &
			\(O\left(sn+kn^3+n^{2k}\right)\) \\
			\hline
		\end{tabular}
		\endgroup
		\caption{Running-time bounds for the three tractable boundary algorithms for
		fixed period \(k\). Here \(n\) is the number of network coordinates and
		\(s\) is the total size of the input coordinate circuits. Since \(k\) is
		fixed, each displayed bound is polynomial in the input size.}
		\label{tab:tractable-complexity}
	\end{table}

	\subsection{Lifting from boundary classes}

	It remains to lift the boundary results to all closed local rule classes. The
	first structural fact is monotonicity with respect to the local rule class. We
	write \(\le_{\classP}\) for polynomial-time many-one
	reducibility~\cite[Chapter~2]{gareyJohnson1979}.

	\begin{lemma}\label{lem:monotonicity}
		Let \(\mathcal F_1\) and \(\mathcal F_2\) be sets of Boolean functions such
		that \(\mathcal F_1 \subseteq \mathcal F_2\). Then
		\[
		k\text{-PLCE}(\mathcal F_1)
		\le_{\classP}
		k\text{-PLCE}(\mathcal F_2).
		\]
	\end{lemma}
	\begin{proof}
		The inclusion \(\mathcal F_1\subseteq\mathcal F_2\) implies
		\([\mathcal F_1]\subseteq[\mathcal F_2]\), because both closures are formed by
		the same closure operations: superposition and introduction of fictitious
		variables.
		By the fixed-basis convention from the input encoding, each source-basis gate has a fixed
		circuit over the target basis. Replacing all source gates by these fixed
		circuits gives an admissible target instance computing the same global map,
		with constant-factor growth in size. The answer is therefore preserved.
	\end{proof}

	The final ingredient is a standard separation statement for closed
	Boolean-function classes.

	\begin{lemma}[closed-class separation~\cite{lau2006}]
		\label{lem:post-separation}
		Let \(\mathcal C\) be a closed class of Boolean functions. If
		\(\mathcal C\) contains none of \(\MajClass\), \(\OrAndClass\), and
		\(\AndOrClass\), then
		\[
		\mathcal C\subseteq \AffClass
		\quad\mathrm{or}\quad
		\mathcal C\subseteq \ConjClass
		\quad\mathrm{or}\quad
		\mathcal C\subseteq \DisjClass.
		\]
	\end{lemma}

	Together with the boundary lemmas, Lemmas~\ref{lem:monotonicity} and
	\ref{lem:post-separation} lift the six boundary results to all closed classes.

	\begin{proof}[Proof of Theorem~\ref{thm:post-dichotomy}]
		Let \(\mathcal C\) be a closed Boolean-function class.
		If \(\mathcal C\) contains \(\MajClass\), \(\OrAndClass\), or \(\AndOrClass\), then
		one of the \(\classNP\)-complete boundary problems reduces to
		\(k\)-PLCE\((\mathcal C)\) by Lemma~\ref{lem:monotonicity}. The target
		problem is in \(\classNP\): since \(k\) is fixed, a candidate configuration
		is verified by computing its first \(k\) iterates. Hence
		\(k\)-PLCE\((\mathcal C)\) is \(\classNP\)-complete.

		If \(\mathcal C\) contains none of \(\MajClass\), \(\OrAndClass\), and
		\(\AndOrClass\), then Lemma~\ref{lem:post-separation} places
		\(\mathcal C\) inside one of the tractable boundary classes \(\AffClass\),
		\(\ConjClass\), or \(\DisjClass\). The fixed-basis substitution used in
		Lemma~\ref{lem:monotonicity} converts the input representation to the
		corresponding boundary basis, after which the boundary algorithm applies.
		Thus
		\(k\)-PLCE\((\mathcal C)\) is polynomial-time solvable. The dichotomy
		follows.
	\end{proof}

	\subsection{Use of AI-assisted tools}

	OpenAI ChatGPT and Codex were used for language editing, stylistic polishing,
	formatting support, and assistance in drafting and refining the TikZ layouts
	for the figures in the article. The authors reviewed and edited the manuscript
	after using these tools and take full responsibility for the content of the
	article.

	\section*{Data availability}

	No datasets were generated or analysed during the current study. The formal
	constructions and algorithms needed to reproduce the results are contained in
	this article.

	\begingroup
	\setlength{\emergencystretch}{2em}
	\printbibliography
	\endgroup

	\section*{Acknowledgements}

	This work was supported by the Ministry of Economic Development of the Russian
	Federation in accordance with the subsidy agreement (agreement identifier
	000000C313925P4H0002; grant No. 139-15-2025-012).

	\section*{Author contributions}

	A.D. contributed to investigation, visualization, and writing--original draft.
	G.B. contributed to methodology, supervision, funding acquisition, and
	writing--review and editing. Both authors reviewed and approved the final
	manuscript.

	\section*{Competing interests}

	The authors declare no competing interests.

\end{document}